\documentclass[letterpaper, 10pt, conference]{ieeeconf}
\IEEEoverridecommandlockouts
\usepackage{cite}
\usepackage{amsmath,amssymb,amsfonts}
\usepackage{algorithmic}
\usepackage{graphicx}
\usepackage{textcomp}
\usepackage{xcolor}
\def\BibTeX{{\rm B\kern-.05em{\sc i\kern-.025em b}\kern-.08em
		T\kern-.1667em\lower.7ex\hbox{E}\kern-.125emX}}

\usepackage{cite}
\usepackage{amsmath,amssymb,amsfonts}
\usepackage{algorithmic}
\usepackage{graphicx}
\usepackage{textcomp}
\usepackage{multirow}

\usepackage{cite}

 \usepackage[subpreambles=true]{standalone}
\usepackage{import}  

\usepackage{tikz}
\usepackage{pgfplots}
\pgfplotsset{compat=newest}
\usetikzlibrary{plotmarks}
\usetikzlibrary{arrows.meta}
\usepgfplotslibrary{patchplots}

\usepackage[hidelinks]{hyperref} 
\usepackage{nicefrac}

\usepackage{makecell}

\usetikzlibrary{shapes,arrows}
\usetikzlibrary{arrows,calc}

\tikzset{
    block/.style = {draw, rectangle, 
        minimum height=1cm, 
        minimum width=2cm},
    input/.style = {coordinate,node distance=1cm},
    output/.style = {coordinate,node distance=2cm},
    arrow/.style={draw, -latex,node distance=2cm},
    pinstyle/.style = {pin edge={latex-, black,node distance=2cm}},
    sum/.style = {draw, circle, node distance=1cm}
}

\newtheorem{define}{Definition}

\newtheorem{theorem}{Theorem}
\newtheorem{lemma}{Lemma}

\newtheorem{cor}{Corollary}

\newtheorem{prop}{Proposition}

\newcommand{\Chat}{\widehat{\mathbf{C}}} 
\newcommand{\Ahat}{\widehat{\mathbf{A}}} 
\newcommand{\Bhat}{\widehat{\mathbf{B}}} 
\newcommand{\I}{\mathbf{I}} 
\newcommand{\mbf}[1]{\mathbf{#1}} 
\newcommand{\bmat}[1]{\begin{bmatrix} #1 \end{bmatrix}} 
\newcommand{\Hcal}{\mathcal{H}} 
\newcommand{\cone}{\mbox{cone}} 
\newcommand{\xhat}{\hat{\mathbf{x}}} 
\newcommand{\A}{\mathbf{A}} 
\newcommand{\B}{\mathbf{B}} 
\newcommand{\C}{\mathbf{C}} 
\newcommand{\D}{\mathbf{D}} 
\newcommand{\x}{\mathbf{x}} 
\newcommand{\bu}{\mathbf{u}} 
\newcommand{\y}{\mathbf{y}} 
\newcommand{\yhat}{\hat{\mathbf{y}}} 
\newcommand{\w}{\mathbf{w}} 
\newcommand{\bP}{\mathbf{P}} 
\newcommand{\Q}{\mathbf{Q}} 
\newcommand{\tr}{\mbox{tr}} 
\newcommand{\K}{\mbf{K}} 
\newcommand{\E}{\mbf{E}} 
\newcommand{\bS}{\mbf{S}} 
\newcommand{\R}{\mbf{R}}
\newcommand{\F}{\mbf{F}}

\newcommand{\M}{\mbf{M}}
\newcommand{\bL}{\mbf{L}}

\newcommand{\X}{\mbf{X}}
\newcommand{\Gcal}{\mathcal{G}}

\newcommand{\st}{\mbox{s.t.}}

\newcommand{\zero}{\mbf{0}}

\newcommand{\Ell}{\mathcal{L}}
\newcommand{\nhat}{\hat{n}}
\newcommand{\uhat}{\hat{\bu}}
\newcommand{\optSpace}{\hspace{3mm}}

\newcommand{\diag}{\mbox{diag}}
\newcommand{\Qbar}{\bar{\Q}}
\newcommand{\Sbar}{\bar{\bS}}
\newcommand{\Rbar}{\bar{\R}}
\newcommand{\br}{\mbf{r}}
\newcommand{\bH}{\mbf{H}}
\newcommand{\e}{\mbf{e}}
\newcommand{\rhat}{\hat{\br}}

\definecolor{mycolor1}{rgb}{0.00000,0.44700,0.74100}%
\definecolor{mycolor2}{rgb}{0.85000,0.32500,0.09800}%
\definecolor{mycolor3}{rgb}{0.92900,0.69400,0.12500}%
\definecolor{mycolor4}{rgb}{0.49400,0.18400,0.55600}%
\definecolor{mycolor5}{rgb}{0.46600,0.67400,0.18800}%
\definecolor{mycolor6}{rgb}{0.30100,0.74500,0.93300}%
\definecolor{mycolor7}{rgb}{0.63500,0.07800,0.18400}%

\usepackage{tikz}
\usepackage{textcomp}
\usepackage{hyperref}
\usepackage{lipsum}

\newcommand\copyrighttext{%
  \footnotesize \textcopyright 2022 IEEE. Personal use of this material is permitted.
  Permission from IEEE must be obtained for all other uses, in any current or future
  media, including reprinting/republishing this material for advertising or promotional
  purposes, creating new collective works, for resale or redistribution to servers or
  lists, or reuse of any copyrighted component of this work in other works.}
\newcommand\copyrightnotice{%
\begin{tikzpicture}[remember picture,overlay]
\node[anchor=south,yshift=10pt] at (current page.south) {\fbox{\parbox{\dimexpr\textwidth-\fboxsep-\fboxrule\relax}{\copyrighttext}}};
\end{tikzpicture}%
}

\begin{document}
	\title{Dissipative Imitation Learning for Robust Dynamic Output Feedback}
	
	\author{Amy Strong$^{1,2}$ and Ethan J. LoCicero$^{1,2}$ and Leila Bridgeman$^{2}$
		\thanks{This material is based upon work supported by the National Science Foundation Graduate Research Fellowship Program under Grant No. 1644868 and by the Alfred P. Sloan Foundation Minority Ph.D. Program.}
		\thanks{$^{1}$ These authors contributed equally to this work.}%
		\thanks{$^{2}$Amy Strong (PhD Student) Ethan J. LoCicero (PhD Candidate) and Leila Bridgeman (assistant Professor) are with the Dept. of Mechanical Eng. and Materials Science at Duke University, Durham, NC, 27708, USA (email: {\tt\small aks121@duke.edu ejl48@duke.edu; ljb48@duke.edu}, phone: 919-660-5310) }%
	}

	\maketitle
	\copyrightnotice
	
	\begin{abstract}
	    Robust imitation learning seeks to mimic expert controller behavior while ensuring stability, but current methods require accurate plant models. Here, robust imitation learning is addressed for stabilizing poorly modeled plants with linear dynamic output feedback. Open-loop input-output properties are used to characterize an uncertain plant, and the feedback matrix of the dynamic controller is learned while enforcing stability through the controller's open-loop QSR-dissipativity properties. The imitation learning method is applied to two systems with parametric uncertainty.
	\end{abstract}
	
	\section{Introduction} \label{sec:intro}
	
	When controller objectives are difficult to formulate due to complex desired behavior or inaccurate system models, imitation learning is an effective alternative. However, unconstrained imitation learning provides no stability guarantees. Recent efforts towards constrained imitation learning provide closed-loop stability or robustness guarantees for nominal linear time invariant (LTI) plant models \cite{Palan2020} \cite{Havens2021}. However, such approaches are limited in cases where accurate models are unavailable. Here, a new dissipativity-based approach provides robust stability guarantees for learned controllers using open-loop input-output (IO) properties of the plant, circumventing the need for accurate state-space models.
	
	Imitation learning bypasses direct controller design and instead uses expert demonstrations of desired system behavior to learn a system's reward function or directly learn a controller \cite{Osa2018}. The expert may be some pre-designed controller or a system's human operator. The simplest form of imitation learning, behavior cloning, is a supervised learning problem in which a mapping from state to action is learned through minimization of a loss function \cite{Osa2018}. While behavioral cloning can create a policy that imitates a stable expert, there are no inherent stability properties of the learned policy. 
	
	In situations that demand stability, robust control theory is being applied to imitation learning \cite{Yin2022,Makdah2021,Palan2020,Havens2021, Pauli2021, Chen2018, Revay2020, Donti2021}. Lyapunov stability theory was used with quadratic constraints to maximize the region of attraction of a closed loop system with neural network feedback control, while minimizing loss \cite{Yin2022}. Similarly, when fitting a policy to expert linear quadratic regulator (LQR) demonstrations, Lipschitz constraints were imposed on loss to ensure stability of feedback control\cite{Makdah2021}. Robust imitation learning has also been applied to linear feedback control policies through incorporation of prior knowledge about the expert demonstrator or the system itself. In \cite{Palan2020}, a Kalman constraint was enforced on the learning process, which assumed that the expert demonstrator was LQR optimal. In \cite{Havens2021}, robust stability was enforced by imposing a threshold on the closed-loop $\text{H}_{\infty}$ norm of a linear plant model during learning. In both cases, a stable policy was learned with a small amount of expert data. While model uncertainty was addressed in \cite{Havens2021}, it may be difficult to select an appropriate closed-loop $\text{H}_{\infty}$ norm if the plant model remains highly uncertain, which is an important use case of learning-based control.
	
	This paper explores the use of IO stability theory to guarantee stability when the LTI plant model is unreliable. In the IO approach, the plant and controller are modeled as mappings from inputs to outputs, and certain open-loop IO properties can be used to infer closed-loop stability. Importantly, these IO properties can often be shown from first principles to hold for nonlinear, time varying, and uncertain parameters, circumventing the problem of unreliable LTI models altogether. Consider the Passivity Theorem \cite{Vidyasagar1977}, which shows that two passive systems in negative feedback are stable. Many nonlinear physical systems are known to be passive for any set of parameters, so even if a passive plant is not well modeled, it must be stabilized by a passive controller \cite{Geromel1997,Brogliato2007}. This approach to stability analysis has been generalized to passivity indices \cite{Vidyasagar1977}, conic sectors \cite{Zames1966}, dissipativity \cite{Hill1977}, and further \cite{Megretski1997,Safonov1980}.
    
    There has been much recent work in designing optimal controllers that are constrained to satisfy desirable IO properties for robust stability. In particular, \cite{Geromel1997,Forbes2019} explore $\Hcal_2$-passive designs, \cite{Sivaranjani2018} applied $\Hcal_2$-conic design to power system stabilization, and \cite{Xia2020,Scorletti2001} developed $\Hcal_\infty$-conic and $\Hcal_\infty$-dissipative designs, respectively. In these designs, an IO property is imposed on the controller as a linear matrix inequality (LMI) constraint during the performance optimization. However, none of these methods have been applied to robust imitation learning. 
    
    Here, the dissipative imitation learning problem is posed for a linear dynamic output feedback controller. The behavior cloning objective is combined with an LMI constraint on the controller that enforces a desired QSR-dissipativity property. This method is similar to that of \cite{Havens2021}, in which a stable linear feedback control was learned for a closed loop system -- referred to later as the Lyapunov-constrained learner. However, in \cite{Havens2021}, an accurate system model was essential to ensuring true stability of the system, whereas here, only coarse input-output knowledge of the system is required to guarantee stability. The QSR-dissipativity framework is chosen because it encompasses passivity, passivity indices, bounded gain, and conic sectors. It can also be used to analyze networks of various IO systems \cite{Vidyasagar1981}. The proposed problem is convex and can be solved efficiently with interior point methods. The resulting controller mimics the behavior of the expert policy while guaranteeing robust closed-loop stability using open-loop IO plant analysis. 
	

	\section{Preliminaries} \label{sec:prelim}
	
	\subsection{Notation} \label{sec:notation}
	
	For a square matrix, $\bP>0$ denotes positive definite. Related properties (negative definiteness and positive/negative semi-definiteness) are denoted likewise. The identity matrix, zero matrix, and trace are $\I$, $\zero$, and $\tr(\cdot)$. Duplicate blocks in symmetric matrices are denoted $(*)$. The $\ell_2$, Frobenius, and $\mathcal{H}_2$ norms are denoted $||\cdot||_2$, $||\cdot||_F$, $||\cdot||_{\mathcal{H}_2}$. Recall $\y \in \Ell_2$ if $||\y||_2^2 = \langle \y,\, \y \rangle = \int_0^\infty \y^T(t)\y(t)dt <\infty$. Further, $\y\in \Ell_{2e}$ if its truncation to $t\in[0,\,T]$ is in $\Ell_2$ $\forall$ $T\geq 0$. The quadruple $(\A,\B,\C,\D)$ denotes the LTI state space $\dot\x = \A\x + \B\bu$, $\y = \C\x + \D\bu$, with states $\x$, inputs $\bu$, and outputs $\y$. The normal distribution with mean $\mu$ and variance $\nu$ is denoted $\mathcal{N}(\mu,\nu)$.
	
	\subsection{Dissipativity and Special Cases} \label{sec:review}
	
	Dissipativity was originally presented by Willems \cite{Willems1972}, and the special case of QSR-dissipativity was soon after defined by Hill and Moylan \cite{Hill1977} for control-affine state-space systems. The following definition for QSR-dissipativity by Vidyasagar is more general and formulated as an IO property to avoid the necessity of a state-space formulation.
	\begin{define}\label{def:dissipative}
		\textit{(Dissipativity \cite{Willems1972,Hill1977})} The operator $\Gcal:\bu\rightarrow \y$ is dissipative with respect to supply rate $w(\bu,\y)$ if for all $\bu\in\Ell_{2e}$ and all $T>0$, $\int_0^T w(t)dt\geq \beta$ for some $\beta\in\mathbb{R}$ depending only on initial conditions. If $w(\bu,\y) = \y^T\Q \y + 2\y^T\bS \bu + \bu^T\R \bu$, then the system is $(\Q,\bS,\R)$-dissipative.
	\end{define}
	Special cases of QSR-dissipativity include conic sectors \cite{Zames1966,Bridgeman2016}, passivity \cite{Brogliato2007}, and bounded gain \cite{Desoer1975}. Their relations to QSR-dissipativity are defined in \autoref{tbl:cases}. Each of these open-loop IO descriptions have an associated stability theorem through which closed-loop IO stability can be established. In $\Ell_2$ space, IO stability is defined as follows.
	\begin{define}
		\textit{(Input-Output or $\Ell_2$ Stability \cite{Vidyasagar1981})} A mapping $\Gcal:\Ell_{2e}\rightarrow\Ell_{2e}$ is $\Ell_2$ stable if any input $\bu\in\Ell_2$ maps to an output $\y\in\Ell_2$.
	\end{define}
	For example, the QSR Stability Theorem below gives conditions on two dissipative systems so that they are IO stable when connected in negative feedback.
	\begin{theorem}\label{thm:QSR}
		\textit{(QSR $\Ell_2$-Stability Theorem \cite{Vidyasagar1981}).} Consider two operators $\Gcal_i:\bu_i\rightarrow\y_i$ that are $(\Q_i,\bS_i,\R_i)$-dissipative for $i=1,2$. Let their negative feedback interconnection be defined as $\bu_1=\br_1-\y_2$ and $\bu_2=\br_2+\y_1$. Then the closed loop from $\br^T=[\br_1^T,\,\br_2^T]$ to $\y^T = [\y_1^T,\,\y_2^T]$ is $\Ell_{2}$ stable if there exists $\alpha>0$ such that 
		\begin{equation*}
			\bmat{\Q_1+\alpha\R_2 & -\bS_1+\alpha\bS_2^T \\ * & \R_1 + \alpha\Q_2} < 0.
		\end{equation*}
	\end{theorem}
	
	The (Extended) Conic Sector Theorem \cite{Zames1966,Bridgeman2016} and the well known Passivity Theorem \cite{Vidyasagar1977} and Small Gain Theorem \cite{Desoer1975} provide IO stability guarantees from similar open-loop IO properties. Briefly, two passive systems in negative feedback are IO stable, two gain-bounded systems in negative feedback are IO stable if their gains multiply to less than one, and the bounds for conic systems have more tedious but similarly simple relationships that guarantee stability.
	
	For LTI systems, variations on the KYP Lemma provide matrix inequality conditions for the satisfaction of IO properties. Two important cases are given below.
	
	\begin{lemma}\label{lem:KYPpass}
		\textit{(Passivity KYP Lemma \cite{Brogliato2007})} Let the system $\Gcal:(\A,\B,\C,\D)$ be controllable and observable. The system is passive if and only if there exists $\bP>0$ such that
		\begin{equation} \label{eqn:KYPpass}
			\bmat{\bP\A+\A^T\bP & \bP\B-\C^T \\ * & -\D-\D^T} \leq 0.
		\end{equation}
	\end{lemma}
	
	\begin{lemma}\label{lem:KYPdiss}
		\textit{(Dissipativity KYP Lemma \cite{Gupta1996})} A square stable LTI system $\Gcal:\Ell_{2e}\rightarrow\Ell_{2e}$ with minimal state space realization $(\A,\B,\C,\D)$ is QSR-dissipative if and only if there exists $\bP>0$ such that 
		\begin{equation} \label{eqn:KYPdiss}
			\bmat{\bP\A{+}\A^T\bP{-}\C^T\Q\C & \bP\B-\C^T(\Q\D+\bS) \\ * & -\R{-}\bS^T\D{-}\D^T\bS{-}\D^T\Q\D} \leq 0.
		\end{equation}
	\end{lemma}
	
	Further, the Network QSR Lemma provides a means of combining local subsystem QSR properties into global QSR properties for an interconnected system.
	
	\begin{lemma}\label{lem:netQSR}
		\textit{(Network QSR Lemma \cite{Vidyasagar1981})}
		Suppose $\Gcal:\br\rightarrow\y$ is composed of $N$ subsystems $\Gcal_i:\bu_i\rightarrow\y_i$, which are $(\Q_i,\bS_i,\R_i)$-dissipative, and let their interconnections be defined by $\bu_i = \br_i - \sum_{i=1}^N \bH_{ij}\y_j$ for disturbance $\br=[\br_1,\dots,\br_N]$, output $\y=[\y_1,\dots,\y_N]$, and interconnection matrix $\bH$. Define $\Q=\diag(\Q_1,\dots,\Q_N)$, $\bS=\diag(\bS_1,\dots,\bS_N)$, and $\R=\diag(\R_1,\dots,\R_N)$. Then $\Gcal$ is $(\Qbar,\Sbar,\Rbar)$-dissipative, where
		$\Qbar = \Q+\bH^T\R\bH - \bS\bH - \bH^T\bS^T, \; \Rbar = \R, \;\Sbar = \bS-\bH^T\R$.
	\end{lemma}
	
	\section{The Case for Dissipativity}
	
	There are two primary reasons for pursuing learning-based design. First, there may be no easily defined objective function that adequately characterizes the desired performance. Second, the plant may be very poorly modeled, in which case traditional objectives like $\Hcal_2$-norm minimization or pole placement would not necessarily yield desirable behavior on the true system. In either case, cloning the behavior of an expert policy circumvents the challenge of posing a useful controller objective. Recent work in achieving stability and robustness guarantees for imitation learning has primarily considered the first perspective, where the plant model can be trusted \cite{Havens2021}. In this case, stability and robustness guarantees on the controller are formulated as closed-loop conditions assuming a nominal LTI plant. The $\Hcal_\infty$ robustness proposed in \cite{Havens2021} may be used to compensate for uncertainty in the model, but determining what bound to use is an open question, especially for complex systems.
	
	When the plant is poorly understood due to parametric uncertainty, unmodeled nonlinearity, delays, etc, closed-loop conditions will not yield reliable stability guarantees. In this case, open-loop conditions based on coarse knowledge of plant IO properties (often derived from first principles) can be used to achieve robust stability guarantees without reliance on accurate state-space models. For example, it is well known that many systems -- such as RLC circuits, PID controllers, and flexible robotic manipulators -- are passive for all possible parameters \cite{Brogliato2007}. This fact follows from physical laws even for nonlinear and time-varying cases. By the Passivity Theorem \cite{Vidyasagar1977}, any passive controller must stabilize such systems. Thus, if an open-loop passivity constraint is imposed on the controller, closed-loop stability is guaranteed without resorting to a deficient LTI model. This stability is guaranteed despite additional noise or inconsistencies in training data, which is especially relevant if a human expert is mimicked. Moreover, while training and test data distributions may vary \cite{Osa2018}, passivity ensures closed-loop stability \cite{Brogliato2007}.

    More generally, dissipativity can be used to convert IO plant information to closed-loop stability guarantees using open-loop controller conditions. At its most general, it is difficult to infer dissipativity from first principles. However, a valuable application of dissipativity is to combine incongruous information from different subsystems. For example, if one subsystem is poorly modeled but known to be passive, another is poorly modeled but has bounded gain, a third is well modeled and known to lie in a particular conic sector, and so on, these conditions can be combined into an overall QSR-dissipativity property according to \autoref{lem:netQSR}. Recent work has also established data-driven methods for identifying QSR properties when models and analytic results are not available \cite{Koch2021AUT}. Thus QSR-dissipativity provides a unified approach to robust stabilization with IO methods and is well suited to the motivations of learning-based control for poorly modeled plants. The next section formalizes the dissipativity-constrained behavior cloning problem.

	\section{Problem Statement} \label{sec:problemstatement}
	
	Consider a plant $\Gcal:\bu\rightarrow\y$ where $\bu,\,\y\in\Ell_{2e}$ are inputs and outputs of dimension $m$ and $p$, respectively. Consider also the LTI control law $\mathcal{C}:(\Ahat,\Bhat,\Chat,\zero)$ with states $\xhat\in\mathbb{R}^{\nhat}$, inputs $\uhat\in\mathbb{R}^p$, and outputs $\yhat\in\mathbb{R}^m$. Let the controller and plant be in negative feedback defined by $\bu=\rhat-\yhat$ and $\uhat = \br +\y$, where $\rhat\in\mathbb{R}^m$ and $\br\in\mathbb{R}^p$ are noise. This dynamic output-feedback control law is composed of observer $(\Ahat,\Bhat)$ and feedback matrix $\Chat$.
	
	Now suppose we have an open-loop stable observer $(\Ahat,\Bhat)$ and an expert policy demonstration defined by a sequence of state-estimate/control-action pairs $\{\xhat_k,\bu_k\}_{k=0}^N$. Importantly, the expert policy $\bu_k$ may have access to better information than the concurrent state estimate $\xhat_k$, but the designed controller will not. The objective is to design the feedback matrix $\Chat$ so that the controller $(\Ahat,\Bhat,\Chat)$ closely imitates the behavior of the expert and satisfies a prescribed QSR-dissipativity condition despite a potentially poorly designed observer. This condition in turn ensures closed-loop stability through an associated IO stability theorem, like \autoref{thm:QSR}. The dissipativity-constrained behavior cloning problem is to minimize $\frac{1}{N}\sum_{k=0}^N l(\Chat\xhat_k,\bu_k) + \eta r(\Chat)$ over $\Chat$ such that $\mathcal{C}:(\Ahat,\Bhat,\Chat,\zero)$ is $(\Q,\bS,\R)$-dissipative, where $l$ is a loss function that empirically measures how well the learned controller mimics the expert policy on the state-estimate data, $r$ is a regularization function to prevent overfitting, and parameter $\eta\in\mathbb{R}_+$ weights the regularization term.
	
	If the plant is poorly modeled, the observer is destined to estimate the states relatively poorly. This is exasperated once the feedback $\Chat$ is added, because the separation principle is lost when $\Ahat$ is fixed \textit{a priori}. However, since the feedback matrix is trained to match the mapping between the (generally bad) state-estimates and the expert control action (which itself does not necessarily rely on the designed observer), the controller can still achieve good performance. Further, as will be shown in the next section, the training problem is always feasible if the observer is open-loop stable.
	
	In keeping with \cite{Havens2021}, a simple and effective choice of loss and regularization are the sum-of-squares, $l(\Chat\xhat_k,\bu_k) = ||\Chat\xhat_k-\bu_k||_2^2$, and the squared Frobenius norm, $r(\Chat) = ||\Chat||_F^2 = \tr(\Chat\Chat^T)$, respectively. \autoref{lem:KYPdiss} can be applied to convert dissipativity condition into a matrix inequality constraint. These choices result in the new problem
	\begin{subequations} \label{eqn:nonconvex_problem}
		\begin{align}
			\min_{\bP>0,\Chat} & \optSpace \frac{1}{N}\sum_{k=0}^N ||\Chat\xhat_k-\bu_k||_2^2 + \eta \tr(\Chat\Chat^T)  \\
			\st & \optSpace \bmat{\bP\Ahat+\Ahat^T\bP-\Chat^T\Q\Chat & \bP\Bhat-\Chat^T\bS \\ * & -\R } \leq 0. \label{eqn:nonconvex_constraint}
		\end{align}
	\end{subequations}
	Equation~\ref{eqn:nonconvex_problem} now has a convex objective, but Constraint~\ref{eqn:nonconvex_constraint} is in general nonlinear. However, for several important special cases, this constraint can be non-conservatively re-posed as an LMI. This is addressed in the next section.

	\section{Main Results} \label{sec:mainresults}
	
	In this section, a new LMI constraint is proposed for imposing QSR-dissipativity with $\Chat$ as the design variable. This new result, established in \autoref{lem:QSRc}, is non-conservative but requires $\Q<0$. This includes bounded gain and interior conic bounds as special cases, while for the special case of passivity, the original constraint is already linear. These special cases are tabulated in terms of their equivalent QSR-dissipativity characterization in \autoref{tbl:cases}. Together, the results of Lemma~\ref{lem:KYPdiss} and Corollary~\ref{lem:QSRc} provide the first unified framework for imposing any interior conic bounds and more generally any QSR property with $\Q<0$ when designing the feedback matrix $\Chat$ for a known observer $(\Ahat,\Bhat)$.
	
	\begin{cor}\label{lem:QSRc}
		Let the LTI system $\Gcal:(\A,\B,\C,\zero)$ be controllable and observable. Then $\Gcal$ is $(\Q,\bS,\R)$-dissipative with $\Q<0$ if and only if there exists $\bP>0$ such that 
		\begin{equation} \label{eqn:QSRc}
			\bmat{\bP\A+\A^T\bP & \bP\B-\C^T\bS & \C^T \\ * & -\R & \zero \\ * & * & \Q^{-1}}\leq 0.
		\end{equation}
	\end{cor}
	\begin{proof} Pull out $-[\C \;\;\; \zero]^T \Q [\C \;\;\; \zero]$ from Constraint~\ref{eqn:nonconvex_constraint}, and apply Schur complement assuming $\Q<0$.
	\end{proof}
	Substituting \autoref{eqn:nonconvex_constraint} for \autoref{eqn:QSRc} or its relevant special case determined by \autoref{tbl:cases} yields the convex optimization
		\begin{equation}\label{eqn:final}
		\hspace{-1mm}	\min_{\bP>0,\Chat} \hspace{1mm} \frac{1}{N}\sum_{k=0}^N ||\Chat\xhat_k-\bu_k||_2^2{+}\eta \tr(\Chat\Chat^T)
			\hspace{1mm} \st \hspace{1mm} \mbox{LMI}(\mbox{Fig.~\ref{tbl:cases}}).
		\end{equation}
	This problem can now be solved efficiently with interior-point methods. The feasibility of Equation~\ref{eqn:final} for interior conic sectors with $a<0<b$ was established in \cite{LoCiceroCDC}, as long as $\Ahat$ is Hurwitz. The more general QSR-dissipative cases in \autoref{tbl:cases} are covered by Proposition~\ref{lem:feas}.
	\begin{prop} \label{lem:feas}
	    Let $\Ahat$ be Hurwitz, $\bS$ be full rank, and $\R\geq0$. Equation~\ref{eqn:final} is feasible with Constraint~\ref{eqn:QSRc} if $\Q<0$, or with Constraint~\ref{eqn:KYPdiss} if $\Q=\zero$. 
	\end{prop} 
	\begin{proof}
	The proof is by construction of a feasible $\Chat$. Consider the first case. If $\bS$ has full rank, then the exists a left pseudo-inverse $\bS^+ = (\bS^T\bS)^{-1}\bS^T$. Let $\Chat^T = \bP\Bhat\bS^+$. Then \autoref{eqn:KYPdiss} becomes $\diag(\bP\Ahat + \Ahat^T\bP + \bP\M\bP,\,-\R) \leq 0$, where $\M = -\Bhat\bS^+\Q{\bS^+}^T\Bhat^T>0$. Front and back multiplying by $\diag(\bP^{-1},\, \I)$ yields $\diag(\Ahat\Pi + \Pi\Ahat^T + \M,\, -\R) \leq 0$, where $\Pi=\bP^{-1}$. Since $\R\geq0$, this is satisfied if $\exists \; \Pi>0$ such that $\Ahat\Pi + \Pi\Ahat^T + \M \leq 0$. Since $\Ahat$ is Hurwitz, so too is $\Ahat^T$, and by Lyapunov's Lemma \cite{Dullerud2005}, this $\Pi$ exists for any $\M>0$. For the second case, let $\Q=\zero$, $\Chat^T = \bP\Bhat\bS^+$. Then \autoref{eqn:KYPdiss} is $\diag(\bP\Ahat + \Ahat^T\bP,\, -\R)\leq 0$. Using the same reasoning, this is satisfied if $\R\geq 0$ and $\Ahat$ is Hurwitz.
	\end{proof}
	
	\begin{figure}
	\centering
		\begin{tabular}{|c | c | c | c|}
			\hline
			Case & QSR & LMI \\
			\hline
			Passive & $\Q=\zero$, $\bS=\frac{1}{2}\I$, $\R=\zero$ & (\ref{eqn:KYPpass}) or (\ref{eqn:KYPdiss}) \\ 
			\hline
			$\gamma$-Bounded Gain & $\Q=-\I$, $\bS=\zero$, $\R=\gamma^2\I$ & (\ref{eqn:QSRc}) \\
			\hline
			\makecell{Nondegenerate \\ Interior Conic \\ $a<0<b$} & \makecell{$\Q=-\I$, $\R=-ab\I$, \\ $\bS=\frac{a+b}{2}\I$} & (\ref{eqn:QSRc}) or \cite{Bridgeman2014} \\
			\hline
			\makecell{Degenerate \\ Interior Conic \\ $d<0$} & \makecell{$\Q=\zero$, $\R=-d\I$, \\ $\bS=\frac{1}{2}\I$} & (\ref{eqn:KYPdiss}) or \cite{Bridgeman2014} \\
			\hline
			\multirow{2}{*}{QSR-dissipative} & $\Q=\zero$, $\R\geq0$, any $\bS$ & (\ref{eqn:KYPdiss}) \\
			& $\Q<0$, $\R\geq0$ any $\bS$ & (\ref{eqn:QSRc})  \\
			\hline
		\end{tabular}
		\caption{Important cases of QSR-dissipativity, their formulations, and the LMI(s) that can be used to impose the property on an LTI system with the output matrix as a free variable.}
		\label{tbl:cases}
	\end{figure}

	\section{Numerical Example} \label{sec:example}
	
	Two experiments are provided to illustrate the utility of the proposed design framework. In each case, LQR-optimal state feedback is used as the expert control policy, and state-estimate feedback based on an LQR-optimal observer is learned according to Section \ref{sec:mainresults}. The first example demonstrates how the proposed design compares to existing robust behavior cloning techniques. A simple passive system is used and the controller's performance in the presence of parametric uncertainty and limited data is explored. The second example demonstrates how QSR-dissipativity can be used to design for networks of nonlinear systems.
	
	\subsection{A Passive System}
	
	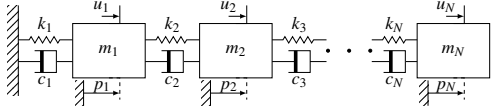
\begin{figure}
	\centering
	\resizebox{.75\columnwidth}{!}{%
		\begin{tikzpicture}[every node/.style={outer sep=0pt},thick,
 mass/.style = {draw,thick},
 spring/.style = {thick,decorate,decoration={zigzag,pre length=0.3cm,post length=0.3cm,segment length=6}},
 ground/.style = {fill,pattern=north east lines,draw=none,minimum width=0.75cm,minimum height=0.3cm},
 dampic/.pic={
 \fill[white] (-0.1,-0.3) rectangle (0.3,0.3);
 \draw (-0.3,0.3) -| (0.3,-0.3) -- (-0.3,-0.3);
 \draw[line width=1mm] (-0.1,-0.3) -- (-0.1,0.3);
 },
roundnode/.style={circle, draw=black, fill=black, scale = 0.25, minimum size=1mm}
 ]

  \node[mass,minimum width=2cm,minimum height=1.5cm] (m1) {\Large $m_1$};
  \node[mass,minimum width=2cm,minimum height=1.5cm,right=1.5cm of
  m1] (m2) {\Large $m_2$};
   \node[mass,minimum width=2cm,minimum height=1.5cm,right=4cm of
  m2] (m3) {\Large $m_{N}$};
  
   \node[left=1.5cm of m1,ground,minimum width=3mm,minimum height=2.5cm] (g1){};
  \draw (g1.north east) -- (g1.south east);
  
  \node[roundnode, right=1.5 of m2] (d1){};
  \node[roundnode, right=0.5cm of d1] (d2){};
  \node[roundnode, right=0.5cm of d2] (d3){};

  \draw[spring] ([yshift=3mm]g1.east) coordinate(aux)
   -- (m1.west|-aux) node[midway,above=1mm]{\Large $k_1$};
  \draw[spring]  (m1.east|-aux) -- (m2.west|-aux) node[midway,above=1mm]{\Large $k_2$};
  \draw[spring]  (m2.east|-aux) -- (d1.west|-aux) node[midway,above=1mm]{\Large $k_3$};
  \draw[spring]  (d3.east|-aux) -- (m3.west|-aux) node[midway,above=1mm]{\Large $k_{N}$};

  \draw ([yshift=-3mm]g1.east) coordinate(aux')
   -- (m1.west|-aux') pic[midway]{dampic} node[midway,below=3mm]{\Large$c_1$}
     (m1.east|-aux') -- (m2.west|-aux') pic[midway]{dampic} node[midway,below=3mm]{\Large $c_2$}
     (m2.east|-aux') -- (d1.west|-aux') pic[midway]{dampic} node[midway,below=3mm]{\Large $c_3$}
     (d3.east|-aux') -- (m3.west|-aux') pic[midway]{dampic} node[midway,below=3mm]{\Large $c_{N}$};

  \foreach \X in {1,2}  
  {\draw[thin] (m\X.north) -| ++ (0.3,0.5) coordinate[pos=.7](aux\X);
   \draw[latex-] (aux\X) -- ++ (-0.5,0) node[above]{\Large $u_\X$}; 
   \draw[thin,dashed] (m\X.south) -| ++ (0.3,-0.6) coordinate[pos=0.85](aux'\X);
   \draw[latex-] (aux'\X) -- ++ (-1,0) node[midway,above,yshift=-1mm]{\Large $p_\X$}
    node[left,ground,minimum height=7mm,minimum width=1mm] (g'\X){};
   \draw[thick] (g'\X.north east) -- (g'\X.south east);
  }
  \draw[thin] (m3.north) -| ++ (0.3,0.5) coordinate[pos=.7](aux3);
   \draw[latex-] (aux3) -- ++ (-0.5,0) node[above]{\Large $u_{N}$}; 
   \draw[thin,dashed] (m3.south) -| ++ (0.3,-0.6) coordinate[pos=0.85](aux'3);
   \draw[latex-] (aux'3) -- ++ (-1,0) node[midway,above,yshift=-1mm]{\Large $p_{N}$}
    node[left,ground,minimum height=7mm,minimum width=1mm] (g'3){};
   \draw[thick] (g'3.north east) -- (g'3.south east);

  \end{tikzpicture}
	}
	\caption{A chain of masses connected by springs and dampers}
	\label{fig:passive}
    \end{figure}

	The plant under consideration is a chain of $N=4$ unit masses connected by springs and dampers, as in \autoref{fig:passive}. 
	Nominal system parameter values are sampled from the uniform distributions $k_n^i\in[k_l,k_h]$ and $c_n^i\in[c_l,c_h],$ while the true system parameters ($\mathbf{k}_t, \mathbf{c}_t)$ are created using nominal system values and percentage of parameter uncertainty. The inputs to the system are forces applied to each mass, and the outputs are velocity measurements, which ensures the system is passive for any parameter set.
	
	The expert policy is designed as a noisy LQR-optimal static state feedback, $\bu_t=-\K_t\x_t + \e$, where $\e$ is noise, $\x_t$ is the state of the true system $(\A_t,\B_t,\C_t,\zero)$, and $\K_t=\E_1^{-1}\B_t^T\Pi_1$, where $\Pi_1$ solves $\A_t^T\Pi_1 + \Pi_1\A_t -\Pi_1\B_t\E_1^{-1}\B_t^T\Pi_1 + \C_t^T\F_1\C_t = \zero$ using $\E_1 = 10$ and $\F_1 =  0.1$. An LQR-optimal observer for the autonomous 
	nominal system, $(\A_n,\B_n,\C_n,\zero)$, is designed as $(\Ahat,\,\Bhat) = (\A_n-\bL\C_n,\, \bL)$, where $\bL=\Pi_`\C_n^T\E_2^{-1}$, and $\Pi_2$ solves $\Pi_2\A_n^T + \A_n\Pi_2 - \Pi_2\C_n^T\E_2^{-1}\C_n\Pi_2 + \B_n\F_2\B_n^T = \zero$ with $\E_2 = 0.5$ and $\F_2 = 0.1$. 
	
	Training data of the expert demonstrator is generated by simulating the expert policy stabilizing the true system from randomized initial system states, while the observer collects state estimates. Initial states are sampled from $\mathcal{N}(\zero, \I)$. Noise ($\mathcal{N}(\zero, 0.25^2\I)$) is added to the controller and plant inputs. Pairs of control actions $\bu_t$ and state estimates $\xhat$ are collected throughout the 10 second duration of each trajectory. Then feedback matrix $\Chat$ is designed using Equation~\ref{eqn:final} with LMI~\ref{eqn:KYPdiss} and $\Q=\zero$, $\bS=\frac{1}{2}\I$, $\R=\zero$, $\eta = 0.05$. The resulting controller $(\Ahat,\Bhat,\Chat,\zero)$ is referred to as the passivity-constrained learner and is passive, as desired.
	
	Performance of the learned controllers is explored in simulations by varying the amount of training data and parametric uncertainty. The passivity-constrained learned controller is compared to an unconstrained learned controller (learned with Equation~\ref{eqn:final} without the constraint) and a Lyapunov-constrained learned controller \cite{Havens2021} which enforces stability through closed-loop state space conditions. Performance is evaluated by implementing  the expert and learned controllers on 100 trajectory simulations in which the initial conditions are randomly sampled from $\mathcal{N}(\zero, 20^2\I),$ well outside of those in the training data. For evaluation, the cost function is defined as, $\frac{1}{N}\sum_{k=0}^N||\x_e(k)-\x(k)||^2_2,$ where $\x_e(k)$ is the expert system states, $\x(k)$ is the system states when using the learned controller, and $N$ is the number of time steps in the simulated trajectory. The plant and control input noise ($\mathcal{N}(\zero, 0.25^2\I)$) is consistent across simulations.
	
	Figure \ref{fig:dataVary} shows the cost of the unconstrained, passivity-constrained, and the Lyapunov-constrained learners \cite{Havens2021} across 50 randomized systems when trained with variable amounts of training trajectories. Nominal system parameters are drawn from the uniform distributions, $k_n^i\in[0.001,10]$ and $c_n^i\in[0.001,1]$. Parametric uncertainty with respect to the true system is 50\%. When trained with one trajectory, the unconstrained learner remains stable for only 70\% of systems. As training trajectories increases, the unconstrained learner then maintains stability for all systems. Both the passivity-constrained and Lyapunov-constrained learners are stable for all amounts of training trajectories. The passivity-constrained learner quickly finds and maintains a consistent performance as training data amount increases. In contrast, the Lyapunov-constrained learner steadily improves in performance. This difference is likely because Equation~\ref{eqn:final} used to find the passivity constrained-learner can be solved using interior point methods, while the Lyapunov-constrained learner \cite{Havens2021} requires projected gradient descent.
	
	Fixing the training data to 25 trajectories, the uncertainty in nominal and true system parameters are then varied from 0 to 100\% for 40 randomized systems. Nominal system parameters are drawn from the uniform distributions $k_n^i\in[0.001,100]$ and $c_n^i\in[0.001,100]$ to represent a large space of parameter uncertainty. The resulting cost and stability results are shown in Figure \ref{fig:paramVary}. Both the passivity-constrained and unconstrained learners exhibit reduced performance with increased parameter variations but remain stable throughout, while the Lyapunov-constrained learner remains stable for only 77.5\% of systems when parameter uncertainty is 25\% and further decreases as uncertainty increases. 
	
    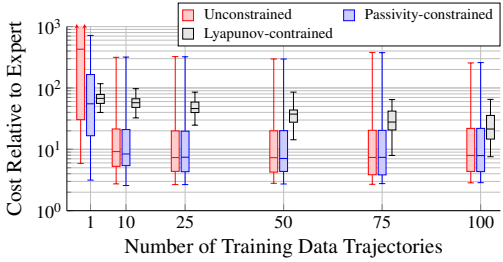
\begin{figure}
	    \centering
	    \resizebox{.8\columnwidth}{!}{%
	%
	%
	\definecolor{mycolor1}{rgb}{0.00000,0.44700,0.74100}%
	\definecolor{mycolor2}{rgb}{0.85000,0.32500,0.09800}%
	\definecolor{mycolor3}{rgb}{0.92900,0.69400,0.12500}%
	\definecolor{mycolor4}{rgb}{0.49400,0.18400,0.55600}%
	\definecolor{mycolor5}{rgb}{0.46600,0.67400,0.18800}%
	\definecolor{mycolor6}{rgb}{0.30100,0.74500,0.93300}%
	\def \w {2}
	\def \fsize {\Large}
	\def \tsize {\large}
	\begin{tikzpicture}
		
		\begin{axis}[
			at={(0in,0in)},
			grid=both,
			width=4in,
			height=1.7in,
			scale only axis,
			bar shift auto,
			xmin=-5,
			xmax=105,
			xtick={1,10,25,50,75,100},
			ymin=1,
			ymax=1000,
			ymode=log,
			xlabel={Number of Training Data Trajectories},
			xlabel style={font=\color{white!15!black},font=\fsize, align=center,yshift=0cm},
			ylabel style={font=\color{white!15!black},font=\fsize, align=center,yshift=0cm},
			ylabel={Cost Relative to Expert},
			tick label style={font=\tsize},
			axis background/.style={fill=white},
			axis x line*=bottom,
			axis y line*=left,
			legend style={yshift=-3.7cm},
			legend entries={Unconstrained, Passivity-constrained, Lyapunov-contrained},
			legend style={/tikz/every even column/.append style={column sep=0.2cm}},
			legend cell align=left,        
			legend style={
				anchor=north west,
				legend columns=2,
				cells={align=left}},
			boxplot/draw direction=y,
			/pgfplots/boxplot/box extend=2,
			legend style={at={(1,1.75)},anchor=south east, column sep=0.5em}
			]
			\addplot[only marks, mark=square*, mark options={scale=1.8, fill=blue}]
			coordinates{ 
				(NaN,NaN)    };
			\addlegendimage{only marks, mark=square*, mark options={scale=1.8, fill=red!20,draw=red}}
			\addplot[only marks, mark=square*, mark options={scale=1.8, fill=red}]
			coordinates{ 
				(NaN,NaN)    };
			\addlegendimage{only marks, mark=square*, mark options={scale=1.8, fill=blue!20,draw=blue}}
			\addplot[only marks, mark=square*, mark options={scale=1.8, fill=blue}]
			coordinates{ 
				(NaN,NaN)    };
			\addlegendimage{only marks, mark=square*, mark options={scale=1.8, fill=gray!20}}
			\addplot+[
			boxplot prepared={
				lower whisker=5.9,
				lower quartile=30.41,
				median=428.17,
				upper quartile=10000,
				upper whisker=10000,
				draw position=-1.5,
				every box/.style={solid,draw=red,fill=red!20},
				every whisker/.style={solid,red},
				every median/.style={solid,red},
			},
			]  coordinates {};
			\addplot+[
			boxplot prepared={
				lower whisker=3.15,
				lower quartile=16.64,
				median=55.2,
				upper quartile=165.38,
				upper whisker=717.16,
				draw position=1,
				every box/.style={solid,draw=blue,fill=blue!20},
				every whisker/.style={solid,blue},
				every median/.style={solid,blue},
			},
			]  coordinates {};
			\addplot+[
			boxplot prepared={
				lower whisker=39.87,
				lower quartile=55.97,
				median=67.64,
				upper quartile=78.35,
				upper whisker=117.31,
				draw position=3.5,
				every box/.style={solid,draw=black,fill=gray!20},
				every whisker/.style={solid,black},
				every median/.style={solid,black},
			},
			]  coordinates {};
						\addplot+[
			boxplot prepared={
				lower whisker=2.73,
				lower quartile=5.28,
				median=9.17,
				upper quartile=21.45,
				upper whisker=315.79,
				draw position=7.5,
				every box/.style={solid,draw=red,fill=red!20},
				every whisker/.style={solid,red},
				every median/.style={solid,red},
			},
			]  coordinates {};
			\addplot+[
			boxplot prepared={
				lower whisker=2.57,
				lower quartile=5.42,
				median=8.39,
				upper quartile=20.97,
				upper whisker=318.82,
				draw position=10,
				every box/.style={solid,draw=blue,fill=blue!20},
				every whisker/.style={solid,blue},
				every median/.style={solid,blue},
			},
			]  coordinates {};
			\addplot+[
			boxplot prepared={
				lower whisker=32.5,
				lower quartile=48.2,
				median=57.45,
				upper quartile=66.7,
				upper whisker=97.29,
				draw position=12.5,
				every box/.style={solid,draw=black,fill=gray!20},
				every whisker/.style={solid,black},
				every median/.style={solid,black},
			},
			]  coordinates {};
			\addplot+[
			boxplot prepared={
				lower whisker=2.66,
				lower quartile=4.37,
				median=7.33,
				upper quartile=19.94,
				upper whisker=324.15,
				draw position=22.5,
				every box/.style={solid,draw=red,fill=red!20},
				every whisker/.style={solid,red},
				every median/.style={solid,red},
			},
			]  coordinates {};
			\addplot+[
			boxplot prepared={
				lower whisker=2.66,
				lower quartile=4.31,
				median=7.43,
				upper quartile=19.75,
				upper whisker=322.11,
				draw position=25,
				every box/.style={solid,draw=blue,fill=blue!20},
				every whisker/.style={solid,blue},
				every median/.style={solid,blue},
			},
			]  coordinates {};
			\addplot+[
			boxplot prepared={
				lower whisker=24.74,
				lower quartile=39.74,
				median=46.15,
				upper quartile=58.86,
				upper whisker=85.53,
				draw position=27.5,
				every box/.style={solid,draw=black,fill=gray!20},
				every whisker/.style={solid,black},
				every median/.style={solid,black},
			},
			]  coordinates {};
			\addplot+[
			boxplot prepared={
				lower whisker=2.79,
				lower quartile=4.25,
				median=7.3,
				upper quartile=20.05,
				upper whisker=296.88,
				draw position=47.5,
				every box/.style={solid,draw=red,fill=red!20},
				every whisker/.style={solid,red},
				every median/.style={solid,red},
			},
			]  coordinates {};
			\addplot+[
			boxplot prepared={
				lower whisker=2.72,
				lower quartile=4.36,
				median=7.08,
				upper quartile=20.12,
				upper whisker=296.3,
				draw position=50,
				every box/.style={solid,draw=blue,fill=blue!20},
				every whisker/.style={solid,blue},
				every median/.style={solid,blue},
			},
			]  coordinates {};
			\addplot+[
			boxplot prepared={
				lower whisker=14.31,
				lower quartile=27.66,
				median=37.46,
				upper quartile=43.73,
				upper whisker=85.47,
				draw position=52.5,
				every box/.style={solid,draw=black,fill=gray!20},
				every whisker/.style={solid,black},
				every median/.style={solid,black},
			},
			]  coordinates {};
			\addplot+[
			boxplot prepared={
				lower whisker=2.69,
				lower quartile=3.85,
				median=7.37,
				upper quartile=20.42,
				upper whisker=379.99,
				draw position=72.5,
				every box/.style={solid,draw=red,fill=red!20},
				every whisker/.style={solid,red},
				every median/.style={solid,red},
			},
			]  coordinates {};
			\addplot+[
			boxplot prepared={
				lower whisker=2.76,
				lower quartile=3.83,
				median=7.39,
				upper quartile=20.37,
				upper whisker=375.82,
				draw position=75,
				every box/.style={solid,draw=blue,fill=blue!20},
				every whisker/.style={solid,blue},
				every median/.style={solid,blue},
			},
			]  coordinates {};
			\addplot+[
			boxplot prepared={
				lower whisker=7.94,
				lower quartile=20.67,
				median=27.83,
				upper quartile=41.94,
				upper whisker=64.64,
				draw position=77.5,
				every box/.style={solid,draw=black,fill=gray!20},
				every whisker/.style={solid,black},
				every median/.style={solid,black},
			},
			]  coordinates {};
			\addplot+[
			boxplot prepared={
				lower whisker=2.84,
				lower quartile=4.36,
				median=7.95,
				upper quartile=21.93,
				upper whisker=255.87,
				draw position=97.5,
				every box/.style={solid,draw=red,fill=red!20},
				every whisker/.style={solid,red},
				every median/.style={solid,red},
			},
			]  coordinates {};
			\addplot+[
			boxplot prepared={
				lower whisker=2.86,
				lower quartile=4.33,
				median=7.9,
				upper quartile=21.94,
				upper whisker=259.05,
				draw position=100,
				every box/.style={solid,draw=blue,fill=blue!20},
				every whisker/.style={solid,blue},
				every median/.style={solid,blue},
			},
			]  coordinates {};
			\addplot+[
			boxplot prepared={
				lower whisker=7.57,
				lower quartile=14.66,
				median=21.4,
				upper quartile=35.66,
				upper whisker=65.2,
				draw position=102.5,
				every box/.style={solid,draw=black,fill=gray!20},
				every whisker/.style={solid,black},
				every median/.style={solid,black},
			},
			]  coordinates {};
		\end{axis}
		\draw [-stealth,color=red](.413,4.33) -- (.413,4.38);
		\draw [-stealth,color=red](.23,4.33) -- (.23,4.38);

	\end{tikzpicture}%
	
	    }
		\caption{Performance of learned controllers for variable amounts of training data trajectories. Systems are designed with $k_n^i\in[0.001,10]$, $c_n^i\in[0.001,1]$, and 50\% parameter uncertainty. Arrows indicate plot extends to infinity.}
		\label{fig:dataVary}
	\end{figure}
	
    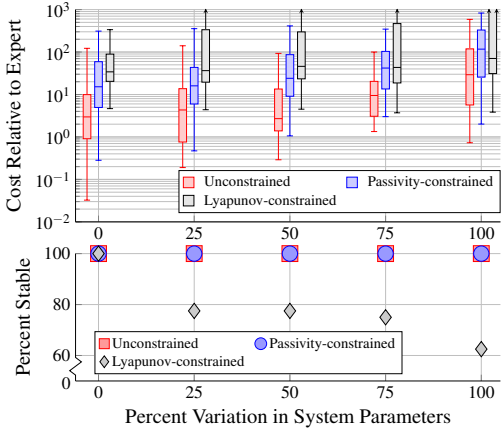
\begin{figure}
	    \centering
	    \resizebox{.8\columnwidth}{!}{%
	%
	%
	\definecolor{mycolor1}{rgb}{0.00000,0.44700,0.74100}%
	\definecolor{mycolor2}{rgb}{0.85000,0.32500,0.09800}%
	\definecolor{mycolor3}{rgb}{0.92900,0.69400,0.12500}%
	\definecolor{mycolor4}{rgb}{0.49400,0.18400,0.55600}%
	\definecolor{mycolor5}{rgb}{0.46600,0.67400,0.18800}%
	\definecolor{mycolor6}{rgb}{0.30100,0.74500,0.93300}%
	\def \w {2}
	\def \fsize {\Large}
	\def \tsize {\large}
	\begin{tikzpicture}
		
		\begin{axis}[
			at={(0in,0in)},
			grid=both,
			width=4in,
			height=2in,
			scale only axis,
			bar shift auto,
			xmin=-6,
			xmax=105,
			xtick={0,25,50,75,100},
			ymin=.01,
			ymax=1000,
			ymode=log,
			ylabel style={font=\color{white!15!black},font=\fsize, align=center,yshift=0cm},
			ylabel={Cost Relative to Expert},
			tick label style={font=\tsize},
			axis background/.style={fill=white},
			axis x line*=bottom,
			axis y line*=left,
			legend style={yshift=-3.85cm},
			legend entries={Unconstrained, Passivity-constrained, Lyapunov-constrained},
			legend style={/tikz/every even column/.append style={column sep=0.2cm}},
			legend cell align=left,        
			legend style={yshift=-.4in,
				anchor=north west,
				legend columns=2,
				cells={align=left}},
			boxplot/draw direction=y,
			/pgfplots/boxplot/box extend=2,
			legend style={at={(1,1)},anchor=south east,legend columns=2,
				column sep=0.5em}
			]
			\addplot[only marks, mark=square*, mark options={scale=1.8, fill=blue}]
			coordinates{ 
				(NaN,NaN)    };
			\addlegendimage{only marks, mark=square*, mark options={scale=1.8, fill=red!20,draw=red}}
			\addplot[only marks, mark=square*, mark options={scale=1.8, fill=red}]
			coordinates{ 
				(NaN,NaN)    };
			\addlegendimage{only marks, mark=square*, mark options={scale=1.8, fill=blue!20,draw=blue}}
			\addplot[only marks, mark=square*, mark options={scale=1.8, fill=blue}]
			coordinates{ 
				(NaN,NaN)    };
			\addlegendimage{only marks, mark=square*, mark options={scale=1.8, fill=gray!20}}
			\addplot+[
			boxplot prepared={
				lower whisker=0.0325,
				lower quartile=0.9092,
				median=2.967,
				upper quartile=9.9462,
				upper whisker=122.44,
				draw position=-3,
				every box/.style={solid,draw=red,fill=red!20},
				every whisker/.style={solid,red},
				every median/.style={solid,red},
			},
			]  coordinates {};
			\addplot+[
			boxplot prepared={
				lower whisker=.282,
				lower quartile=4.95,
				median=15.21,
				upper quartile=59.28,
				upper whisker=311.13,
				draw position=0,
				every box/.style={solid,draw=blue,fill=blue!20},
				every whisker/.style={solid,blue},
				every median/.style={solid,blue},
			},
			]  coordinates {};
						\addplot+[
			boxplot prepared={
				lower whisker=4.65,
				lower quartile=20.39,
				median=34.18,
				upper quartile=89.27,
				upper whisker=336.74,
				draw position=3,
				every box/.style={solid,draw=black,fill=gray!20},
				every whisker/.style={solid,black},
				every median/.style={solid,black},
			},
			]  coordinates {};
			\addplot+[
			boxplot prepared={
				lower whisker=0.19,
				lower quartile=0.76,
				median=4.32,
				upper quartile=13.71,
				upper whisker=140.55,
				draw position=22,
				every box/.style={solid,draw=red,fill=red!20},
				every whisker/.style={solid,red},
				every median/.style={solid,red},
			},
			]  coordinates {};
			\addplot+[
			boxplot prepared={
				lower whisker=.47,
				lower quartile=5.98,
				median=16,
				upper quartile=43.5,
				upper whisker=355.5,
				draw position=25,
				every box/.style={solid,draw=blue,fill=blue!20},
				every whisker/.style={solid,blue},
				every median/.style={solid,blue},
			},
			]  coordinates {};
			\addplot+[
			boxplot prepared={
				lower whisker=4.38,
				lower quartile=19.53,
				median=36.49,
				upper quartile=335.62,
				upper whisker=10000,
				draw position=28,
				every box/.style={solid,draw=black,fill=gray!20},
				every whisker/.style={solid,black},
				every median/.style={solid,black},
			},
			]  coordinates {};
			\addplot+[
			boxplot prepared={
				lower whisker=0.29,
				lower quartile=1.39,
				median=2.7,
				upper quartile=13.46,
				upper whisker=93,
				draw position=47,
				every box/.style={solid,draw=red,fill=red!20},
				every whisker/.style={solid,red},
				every median/.style={solid,red},
			},
			]  coordinates {};
			\addplot+[
			boxplot prepared={
				lower whisker=1.06,
				lower quartile=9.09,
				median=24.08,
				upper quartile=86.79,
				upper whisker=413.31,
				draw position=50,
				every box/.style={solid,draw=blue,fill=blue!20},
				every whisker/.style={solid,blue},
				every median/.style={solid,blue},
			},
			]  coordinates {};
			\addplot+[
			boxplot prepared={
				lower whisker=4.51,
				lower quartile=23.38,
				median=46.04,
				upper quartile=298.28,
				upper whisker=10000,
				draw position=53,
				every box/.style={solid,draw=black,fill=gray!20},
				every whisker/.style={solid,black},
				every median/.style={solid,black},
			},
			]  coordinates {};
			\addplot+[
			boxplot prepared={
				lower whisker=1.34,
				lower quartile=3.07,
				median=9.48,
				upper quartile=20.39,
				upper whisker=100.13,
				draw position=72,
				every box/.style={solid,draw=red,fill=red!20},
				every whisker/.style={solid,red},
				every median/.style={solid,red},
			},
			]  coordinates {};
			\addplot+[
			boxplot prepared={
				lower whisker=3.01,
				lower quartile=13.5,
				median=41.96,
				upper quartile=103.59,
				upper whisker=346.35,
				draw position=75,
				every box/.style={solid,draw=blue,fill=blue!20},
				every whisker/.style={solid,blue},
				every median/.style={solid,blue},
			},
			]  coordinates {};
			\addplot+[
			boxplot prepared={
				lower whisker=3.7,
				lower quartile=18.71,
				median=43.39,
				upper quartile=468.12,
				upper whisker=10000,
				draw position=78,
				every box/.style={solid,draw=black,fill=gray!20},
				every whisker/.style={solid,black},
				every median/.style={solid,black},
			},
			]  coordinates {};
			\addplot+[
			boxplot prepared={
				lower whisker=.73,
				lower quartile=5.67,
				median=29.27,
				upper quartile=118.57,
				upper whisker=588.5,
				draw position=97,
				every box/.style={solid,draw=red,fill=red!20},
				every whisker/.style={solid,red},
				every median/.style={solid,red},
			},
			]  coordinates {};
			\addplot+[
			boxplot prepared={
				lower whisker=2.01,
				lower quartile=25.75,
				median=117.08,
				upper quartile=330.2,
				upper whisker=830.05,
				draw position=100,
				every box/.style={solid,draw=blue,fill=blue!20},
				every whisker/.style={solid,blue},
				every median/.style={solid,blue},
			},
			]  coordinates {};
			\addplot+[
			boxplot prepared={
				lower whisker=3.84,
				lower quartile=31.09,
				median=70.6,
				upper quartile=10000,
				upper whisker=10000,
				draw position=103,
				every box/.style={solid,draw=black,fill=gray!20},
				every whisker/.style={solid,black},
				every median/.style={solid,black},
			},
			]  coordinates {};
		\end{axis}
	\draw [-stealth](5.4,5.05) -- (5.4,5.1);
	\draw [-stealth](3.11,5.05) -- (3.11,5.1);
	\draw [-stealth](7.69,5.05) -- (7.69,5.1);
	\draw [-stealth](9.885,5.05) -- (9.885,5.1);
	\draw [-stealth](10.065,5.05) -- (10.065,5.1);
		
		\begin{axis}[%
			width=4in,
			axis y discontinuity=crunch,
			height=1.2in,
			at={(0in,-1.5in)},
			scale only axis,
			xmin=-6,
			xmax=105,
			xtick={0,25,50,75,100},
			ytick={0,60,80,100},
			xlabel style={font=\color{white!15!black},font=\fsize},
			xlabel={Percent Variation in System Parameters},
			ymin=50,
			ymax=100,
			yminorticks=true,
			grid,
			ylabel style = {font=\fsize,yshift=.09cm},
			ylabel={Percent Stable},
			tick label style={font=\tsize},
			axis background/.style={fill=white},
			axis x line*=bottom,
			axis y line*=left,
			legend style={yshift=-1.8cm, xshift=-9.5cm},
			legend entries={Unconstrained, QSR-constrained, Havens},
			legend style={/tikz/every even column/.append style={column sep=0.2cm}},
			legend cell align=left,        
			legend style={yshift=0in,
				anchor=north west,
				legend columns=2,
				cells={align=left}},
			]
			\addplot[only marks, mark=square*, mark options={scale=2.6, fill=blue}]
			coordinates{ 
				(NaN,NaN)    };
			\addlegendimage{only marks, mark=square*, mark options={scale=1.7, fill=red!40,draw=red}}
			\addplot[only marks, mark=square*, mark options={scale=2.6, fill=red}]
			coordinates{ 
				(NaN,NaN)    };
			\addlegendimage{only marks, mark=*, mark options={scale=2, fill=blue!40,draw=blue}}
			\addplot[only marks, mark=square*, mark options={scale=2.6, fill=blue}]
			coordinates{ 
				(NaN,NaN)    };
			\addlegendimage{only marks, mark=diamond*, mark options={scale=2, fill=gray!40}}
			\addplot[only marks,mark=square*,mark options={scale=2.7,fill=red!40,draw=red}] table[row sep=crcr] {%
				0 100 \\
				25 100 \\
				50 100 \\
				75 100 \\
				100 100\\
			};
			\addlegendentry{Unconstrained}
			\addplot[only marks,mark=*,mark options={scale=2.6,fill=blue!40,draw=blue}] table[row sep=crcr] {%
				0 100 \\
				25 100 \\
				50 100 \\
				75 100 \\
				100 100\\
			};
			\addlegendentry{Passivity-constrained}
			\addplot[only marks,mark=diamond*,mark options={scale=2.6,fill=gray!40,draw=black}] table[row sep=crcr] {%
				0 100 \\
				25 77.5\\
				50 77.5\\
				75 75\\
				100 62.5\\
			};
			\addlegendentry{Lyapunov-constrained}
		\end{axis}
		\node[] at (-.09in,-1.52in) {\large 0};
		
	\end{tikzpicture}%
	
	    }
	    \caption{Performance of learned controllers across variable system parameter uncertainty when trained with 25 training data trajectories. Nominal parameter values are drawn from $k_n^i\in[0.001,100]$ and $c_n^i\in[0.001,100]$. Arrows indicate plot extends to infinity.}
	    \label{fig:paramVary}
    \end{figure}

	\subsection{A QSR System}
	
	Consider $\Gcal_t:(u_1,\, u_2) \rightarrow(y_1,\, y_2)$, the interconnection of two subsystems $\Gcal_1^t:e_1\rightarrow y_1$ and $\Gcal_2^t:e_2\rightarrow y_2$ defined by
	\begin{align*}
		\Gcal^t_1 : \begin{cases}
		\dot x_1 = -x_1^3-x_1+e_1 \\
		y_1 = \dot x_1-2e_1 \\
		e_1 = u_1+y_2,
	    \end{cases} &
	    \Gcal^t_2 : \begin{cases}
	    \dot \x_2 = \A_{2t}\x_2 + \B_{2t} e_2 \\
	    y_2 = \C_{2t} \x_2 \\
	    e_2 = u_2-y_1,
	    \end{cases}
	\end{align*}
	where $\Gcal^t_2$ is the same system as in the prior experiment with two masses and only one input and output, applied to the first mass. The system parameters are $m_1,m_2 = 0.5$, $k_1,k_2 = 5$, $c_1,c_2 = 10^{-3}$. The expert control policy is $\bu_t = -\K_t\x_t+\mathbf{e}$, where $\mathbf{e}$ is noise, $\x_t^T = [x_1, \x_2^T]$, $\bu_t^T = [u_1, u_2]$, and $\K_t$ is LQR-optimal state feedback gain for $\Gcal_t$ linearized about $(x_1,e_1) = (0,0)$. This is designed similarly to the previous example with parameters $\E_1 = 1000\I$ and $\mathbf{F}_1=\I$. 
	
	For the purposes of control design, $\Gcal_1^t$ is known perfectly but the nominal system is linearized about $(x_1,e_1) = (0,0)$. Meanwhile, $\Gcal^t_2$ is a poorly understood subsystem. To reflect the lack of modeling information, it is estimated as a single-mass system with nominal parameters $m_1=1$, $k_1 = 2.5$, $c_1 = 0.05$. This reflects a well measured lumped mass and spring constant and overestimated damping. The resulting nominal system is $\Gcal_n:(\A_n,\B_n,\C_n,\D_n)$, defined by subsystems
		\begin{align*}
		\Gcal_1^n : \begin{cases}
		\dot x_1 = -x_1+e_1 \\
		y_1 = -x_1-e_1\\
		e_1 = u_1+y_2,
	    \end{cases} &
	    \Gcal_2^n : \begin{cases}
	    \dot \x_2 = \A_{2n}\x_2 + \B_{2n} e_2 \\
	    y_2 = \C_{2n} \x_2\\
	    e_2 = u_2-y_1.
	    \end{cases}
	\end{align*}
	
	Since the learned controller does not have direct access to states, an LQR-optimal observer of the form $(\Ahat,\,\Bhat) = (\A_n-\bL\C_n,\,\bL)$ is designed, where $\bL$ is designed as in the previous section with $\E_2 = 10\I$ and $\mathbf{F}_2 = \I$. Since the nonlinear $\Gcal^t_1$ is known perfectly, it is known that $\Gcal^t_1\in\cone[-2,-1]$ \cite{Bridgeman2016}. Although $\Gcal^t_2$ is modeled poorly, the true system is known to be passive from first principles. Applying \autoref{lem:netQSR}, $\Gcal_t$ is QSR-dissipative with parameters
	\begin{equation*}
	    \Qbar = \bmat{-1 & 1 \\ 1 & -2}, \optSpace \Sbar = \bmat{-\frac{3}{2} & 0 \\ 2 & \frac{1}{2}}, \optSpace \Rbar = \bmat{-2 & 0 \\ 0 & 0}
	\end{equation*}
	from $\bu = [u_1,\; u_2]^T$ to $\y = [y_1,\; y_2]^T$. Then applying \autoref{thm:QSR}, any controller stabilizes the true system if it is QSR-dissipative with respect to
	\begin{equation*}
	    \Q_c {=} \bmat{-0.52 & 0 \\ 0 & -1.04},\,\bS_c {=} \bmat{-\frac{3}{2} & 2 \\ 0 & \frac{1}{2}},\, \R_c {=} \bmat{0.45 & -.48 \\ * & 0.92}.
	\end{equation*}
	
	Two learned controllers are designed. The first, $\mathcal{C}_{QSR}:(\Ahat,\Bhat,\Chat_{QSR},\zero)$, is designed using dissipativity-constrained behavior cloning as outlined in \autoref{sec:mainresults}, where the controller is constrained to be $(\Q_c,\bS_c,\R_c)$-dissipative. The second, $\mathcal{C}_{nc}:(\Ahat,\Bhat,\Chat_{nc},\zero)$, is designed using unconstrained behavior cloning, which is achieved by solving \autoref{eqn:final} without the constraint. In both cases, $\eta = 0.05$, and training data with the expert controller is generated using 15-second trajectories with uniformly distributed initial conditions within $||\x_0||\leq5$, $||\xhat_0||\leq 0.25$. White noise with distribution $\mathcal{N}(\zero,0.25^2\I)$ was added to plant and controller inputs to represent environmental disturbances. 

	To evaluate the performance of $\mathcal{C}_{QSR}$ and $\mathcal{C}_{nc}$, the two learned controllers and the expert were simulated over twenty-five 15-second trajectories. Initial conditions for the simulations were uniformly distributed within $||\x_0||\leq 20$, $||\xhat_0||\leq 1$ and were constant across the three controller simulations. The environmental noise distribution for the simulation was $\mathcal{N}(\zero,\I)$ and was also constant across simulations. The increased variability of the noise and initial conditions in the test is used to reflect performance outside of the training data set. The performance of the learned controllers was then compared to that of the expert through the cost function $\frac{1}{N}\sum_{k=0}^N||\x_e(k)-\x(k)||^2_2$ for each trajectory, where $x_e(k)$ is the state at time step $k$ given the expert control action, and $x(k)$ is the state given the learned controller action. \autoref{fig:Ethanplot} shows the relative cost for each controller as a function of the number of training data trajectories. The QSR-constrained controller significantly outperformed the unconstrained controller in very low-data settings, and provided comparable performance in higher-data settings. This is because the unconstrained controller demonstrated unstable behavior in many low-data simulations, while the QSR-constrained controller was stable for all.

    \begin{figure}
	    \centering
	    \resizebox{.8\columnwidth}{!}{%
	    	\subimport{tikz_images/}{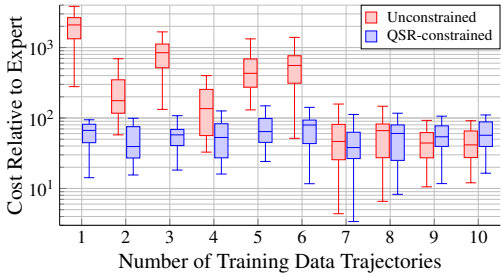}
	    }
	    \caption{Performance of QSR-constrained behavior cloning controller compared to the unconstrained controller as applied to Example 2.}
	    \label{fig:Ethanplot}
    \end{figure}

	\section{Discussion} \label{sec:discussion}
	
	In this work, it is shown that dissipativity is a valuable tool for recovering robust stability guarantees for learned controllers when plant models are low-fidelity. The behavior cloning control problem is reformulated for observer-based dynamic output feedback control, and convex constraints are formulated to enforce QSR-dissipative properties on the learned controller. IO stability theory is then employed to guarantee stability of the true system using coarse open-loop knowledge of the plant subsystems and their interconnections. Experimental results show that this approach yields controllers that are capable of imitating expert behavior despite limited information about the system and its true states. Moreover, the constrained learned controllers maintain stability despite high uncertainty and low available data, even in cases where comparable methods fail due to over-reliance on unreliable plant models. 
	Future work will improve on this scheme by designing the observer and feedback together, relaxing constraints on the possible QSR-dissipative properties that can be imposed, and applying dissipative constraints on different forms of learned controllers, like neural networks.

	\addtolength{\textheight}{0cm}  
	\bibliographystyle{IEEEtran}
	\bibliography{IEEEabrv,CCTA2022Citations}

\begin{thebibliography}{10}
\providecommand{\url}[1]{#1}
\csname url@samestyle\endcsname
\providecommand{\newblock}{\relax}
\providecommand{\bibinfo}[2]{#2}
\providecommand{\BIBentrySTDinterwordspacing}{\spaceskip=0pt\relax}
\providecommand{\BIBentryALTinterwordstretchfactor}{4}
\providecommand{\BIBentryALTinterwordspacing}{\spaceskip=\fontdimen2\font plus
\BIBentryALTinterwordstretchfactor\fontdimen3\font minus
  \fontdimen4\font\relax}
\providecommand{\BIBforeignlanguage}[2]{{%
\expandafter\ifx\csname l@#1\endcsname\relax
\typeout{** WARNING: IEEEtran.bst: No hyphenation pattern has been}%
\typeout{** loaded for the language `#1'. Using the pattern for}%
\typeout{** the default language instead.}%
\else
\language=\csname l@#1\endcsname
\fi
#2}}
\providecommand{\BIBdecl}{\relax}
\BIBdecl

\bibitem{Palan2020}
M.~Palan, S.~Barratt, A.~McCauley, D.~Sadigh, V.~Sindhwani, and S.~Boyd,
  ``Fitting a linear control policy to demonstrations with a kalman
  constraint,'' in \emph{Proc. 2nd Conf. Learn. Dyn. Control (L4DC)}, ser.
  Proc. Mach. Learn, vol. 120, 10--11 Jun 2020, pp. 374--383.

\bibitem{Havens2021}
A.~Havens and B.~Hu, ``On estimation learning of linear control policies:
  enforcing stability and robustness constraints via {LMI} conditions,''
  \emph{Amer. Ctrl. Conf.}, pp. 882--887, May 2021.

\bibitem{Osa2018}
T.~Osa, J.~Pajarinen, G.~Neumann, J.~A. Bagnell, P.~Abbeel, and J.~Peters,
  \emph{An Algorithmic Perspective on Imitation Learning}.\hskip 1em plus 0.5em
  minus 0.4em\relax Found. and Trends in Robotics, 2017, vol.~7, no. 1-2.

\bibitem{Yin2022}
H.~Yin, P.~Seiler, M.~Jin, and M.~Arcak, ``Imitation learning with stability
  and safety guarantees,'' \emph{IEEE Control Syst. Lett.}, vol.~6, pp.
  409--414, 2022.

\bibitem{Makdah2021}
A.~A.~A. Makdah, V.~Krishnan, and F.~Pasqualetti, ``Learning robust feedback
  policies from demonstrations,'' \emph{Ar$\chi$iv}, March 2021.

\bibitem{Pauli2021}
P.~Pauli, J.~K\"ohler, J.~Berberich, A.~Koch, and F.~Allg\"ower, ``Offset-free
  setpoint tracking using neural network controllers,'' in \emph{Proc. 3rd
  Conf. Learn. Dyn. Control (L4DC)}, ser. Proc. Mach. Learn, vol. 144, 07 -- 08
  June 2021, pp. 992--1003.

\bibitem{Chen2018}
S.~Chen, K.~Saulnier, N.~Atanasov, D.~D. Lee, V.~Kumar, G.~J. Pappas, and
  M.~Morari, ``Approximating explicit model predictive control using
  constrained neural networks,'' in \emph{Amer. Ctrl. Conf.}, 2018, pp.
  1520--1527.

\bibitem{Revay2020}
M.~Revay and I.~Manchester, ``Contracting implicit recurrent neural networks:
  Stable models with improved trainability,'' in \emph{Proc. 2nd Conf. Learn.
  Dyn. Control (L4DC)}, ser. Proc. Mach. Learn, vol. 120, 10--11 Jun 2020, pp.
  393--403.

\bibitem{Donti2021}
P.~L. Donti, M.~Roderick, M.~Fazlyab, and J.~Z. Kolter, ``Enforcing robust
  control guarantees within neural network policies,'' in \emph{Int. Conf.
  Learn. Represent.}, 2021.

\bibitem{Vidyasagar1977}
M.~Vidyasagar, ``{$\Ell_2$}-stability of interconnected systems using a
  reformulation of the passivity theorem,'' \emph{IEEE Tran. Circ. Sys.}, vol.
  cas-24, no.~11, pp. 637--645, Nov. 1977.

\bibitem{Geromel1997}
J.~C. Geromel and P.~B. Gapski, ``Synthesis of positive real {$\mathcal{H}_2$
  controllers},'' \emph{IEEE Tran. Aut. Ctrl.}, vol.~42, no.~7, pp. 988--992,
  Jul. 1997.

\bibitem{Brogliato2007}
B.~Brogliato, R.~Lozano, B.~Maschke, and O.~Egeland, \emph{Dissipative Systems
  Analysis and Control: Theory and Applications}, 2nd~ed.\hskip 1em plus 0.5em
  minus 0.4em\relax London, UK: Springer Verlag, 2007.

\bibitem{Zames1966}
G.~Zames, ``On the input-output stability of time-varying nonlinear feedback
  systems parts {I \& II},'' \emph{IEEE Tran. Aut. Ctrl.}, vol. \textsc{ac}-11,
  no. 2--3, Apr./Jul. 1966.

\bibitem{Hill1977}
D.~J. Hill and P.~J. Moylan, ``Stability results for nonlinear feedback
  systems,'' \emph{Automatica}, vol.~13, no.~4, pp. 377--382, 07 1977.

\bibitem{Megretski1997}
A.~Megretski and A.~Rantzer, ``System analysis via integral quadratic
  constraints,'' \emph{IEEE Tran. Aut. Ctrl.}, vol.~42, no.~6, pp. 819--830,
  1997.

\bibitem{Safonov1980}
Safonov, \emph{Stability and Robustness of multivariable feedback
  systems}.\hskip 1em plus 0.5em minus 0.4em\relax Cambridge, MA: MIT Press,
  1980.

\bibitem{Forbes2019}
J.~R. Forbes, ``Synthesis of strictly positive real {$\Hcal_2$} controllers
  using dialated {LMI}s,'' \emph{Int. J. Ctrl.}, vol.~92, no.~11, pp.
  2584--2590, 2019.

\bibitem{Sivaranjani2018}
S.~Sivaranjani, J.~R. Forbes, P.~Seiler, and V.~Gupta, ``Conic-sector-based
  analysis and control synthesis for linear parameter varying systems,''
  \emph{IEEE Tran. Aut. Ctrl.}, vol.~2, no.~2, Apr. 2018.

\bibitem{Xia2020}
M.~Xia, P.~Gahinet, N.~Abroug, C.~Buhr, and E.~Laroche, ``Sector bounds in
  stability analysis and control design,'' \emph{Int. J. Robust Nonlin. Ctrl.},
  vol.~30, pp. 7857--7882, May 2020.

\bibitem{Scorletti2001}
G.~Scorletti and G.~Duc, ``An {LMI} approach to decentralized {$H_{\infty}$}
  control,'' \emph{Int. J. Ctrl.}, vol.~74, no.~3, pp. 211--224, 2001.

\bibitem{Vidyasagar1981}
M.~Vidyasagar, \emph{Input-output analysis of large-scale interconnected
  systems}.\hskip 1em plus 0.5em minus 0.4em\relax Springer-Verlag, 1981.

\bibitem{Willems1972}
J.~C. Willems, ``Dissipative dynamical systems part {I}: General theory,''
  \emph{Arch. Rational Mech. Anal.}, vol.~45, pp. 321--351, 1972.

\bibitem{Bridgeman2016}
L.~J. Bridgeman and J.~R. Forbes, ``The extended conic sector theorem,''
  \emph{IEEE Tran. Aut. Ctrl.}, vol.~61, no.~7, pp. 1931--1937, Jul. 2016.

\bibitem{Desoer1975}
C.~A. Desoer and M.~Vidyasagar, \emph{Feedback Systems: Input-Output
  Properties}, 2nd~ed.\hskip 1em plus 0.5em minus 0.4em\relax New York, NY:
  Academic Press, Inc., 1975.

\bibitem{Gupta1996}
S.~Gupta, ``Robust stabiliation of uncertain systems based on energy
  dissipation concepts,'' Vigyan, Inc., Hampton, VA, Tech. Rep. NASA Contractor
  Report 4713, 1996.

\bibitem{Koch2021AUT}
A.~Koch, J.~Berberich, K{\"{o}}hler, and Allg{\"{o}}wer, ``Determining optimal
  input-output properties: A data driven approach,'' \emph{Automatica}, vol.
  134, pp. 1--13, 2021.

\bibitem{LoCiceroCDC}
E.~J. LoCicero and L.~Bridgeman, ``Fixed-order {$\mathcal{H}_2$}-conic
  control,'' \emph{66th IEEE Conf. Decis. Ctrl.}, 2021, to be published.

\bibitem{Dullerud2005}
G.~E. Dullerud and F.~G. Paganini, \emph{A course in robust control theory: A
  convex approach}.\hskip 1em plus 0.5em minus 0.4em\relax New York, NY:
  Springer, 2005, vol.~36.

\bibitem{Bridgeman2014}
L.~J. Bridgeman and J.~R. Forbes, ``Conic-sector-based control to circumvent
  passivity violations,'' \emph{Int. J. Ctrl.}, vol.~87, no.~8, pp. 1467--1477,
  2014.

\end{thebibliography}
	
\end{document}